\documentclass{article}

\usepackage[english]{babel}

\usepackage[framemethod=TikZ]{mdframed} 
\usepackage{xcolor}

\definecolor{softblue}{rgb}{0.122, 0.435, 0.698}

\newmdenv[
  linewidth=1pt, 
  roundcorner=4pt,
  linecolor=softblue,
  innerleftmargin=6pt,
  innerrightmargin=6pt,
  innertopmargin=6pt,
  innerbottommargin=6pt
]{refereebox}

\usepackage[letterpaper,top=2cm,bottom=2cm,left=3cm,right=3cm,marginparwidth=1.75cm]{geometry}

\usepackage{datetime}
\usepackage{csquotes}
\usepackage{amsmath}
\usepackage{authblk}
\usepackage{amsthm}
\usepackage{amssymb}
\usepackage[colorlinks=true, allcolors=blue]{hyperref}
\usepackage[backend=biber,style=numeric]{biblatex}
\newcounter{common}[section]

\newtheorem{definition}[common]{Definition}

\newtheorem{theorem}[common]{Theorem}

\title{Gravitational Field of a Rotating Mass on an Expanding Universe}
\author{Antonio Peña Peña\footnote{Contact email: \href{mailto:antoniopenapena52@gmail.com}{antoniopenapena52@gmail.com}.}}
\affil{Higher Polytechnic School, Autonomous University of Madrid, 28049-Cantoblanco, Spain}

\newdate{date}{24}{07}{2024}
\date{\displaydate{date}}

\begin{document}
\maketitle

\begin{abstract}
We present a new exact solution to Einstein's field equations that unifies the Kerr black hole with Friedmann-Lemaître-Robertson-Walker cosmology. This metric reduces to Kerr-de Sitter in the appropriate limit and accounts for cosmological expansion dynamics both inside and outside the black hole. The model predicts a stationary mass, as well as a contracting ergosphere and event horizon with respect to the expanding cosmic rest frame. This result correctly generalises the McVittie metric to rotating masses or Kerr-de Sitter to arbitrary scale factors $a(t)$. Additionally, we find that the ergosphere tends to fade away with respect to the universe's expansion and no further interactions of dark energy with respect to the black hole's rotation.
\end{abstract}

\noindent The quest for a description of black holes (BHs) in Friedmann-Lemaître-Robertson–Walker (FLRW) cosmology spans nearly over a century, yet efforts towards this goal have encountered limited success. As noted by Farrah \textit{et al.} \cite{FarrahEtal}, early works date back to 1933 when McVittie \cite{McVittie} extended the Schwarzschild solution to FLRW spacetimes. A further study by Nolan \cite{Nolan} used MacVittie's result to build a non-singular solution for its interior, and other findings from Faraoni \& Jacques \cite{FaraoniJacques} feature interior time-varying mass, energy and pressure densities, as well as event horizons that comove with the universe's expansion. Yet even with these great successes, these models do not describe spinning BHs and are therefore incompatible with Kerr spacetimes. And even if they were, as Faraoni \& Jacques note, it is still difficult to interpret McVittie's and Nolan solutions as BHs because they are singular on the 2-sphere at $r = m/2$ (in natural units).

The Kerr \cite{Kerr} solution to Einstein's field equations describe to almost perfect precision the evolution and behaviour of rotating BHs according to observations such as the detection of gravitational waves from merging binary BHs \cite{AbbottEtal2019,AbbottEtal2021} and the direct imaging of supermassive BHs by the Event Horizon Telescope \cite{Akiyama2019}. This data confirms the Kerr metric for relatively short timescales (from milliseconds to hours) and considerably wide spatial scales (up to milliparsecs). 

Nevertheless, even though the Kerr metric describes BHs in a magnificent way on short timescales, it was built upon some boundary conditions that may not be physically representative for our universe \cite{WiltshireEtal}. Namely, the assumption of an asymptotically flat spacetime does not concord with an expanding universe following the rules of FLRW cosmology, which has also been confirmed to very high precision \cite{AghanimEtal,DodelsonSchmidt}. This disparity might not seem important at first glance, especially given the fact that spatial expansion is negligible over not-too-long timescales (like the BHs interior and vicinity); but in reality they are key to understanding the long-term evolution of these dark entities. Recently, Farrah \textit{et al.} \cite{FarrahEtal} posed the question that because we still lack a solution linking local Kerr to the FLRW limit, it is in theory possible for BHs to be the source of dark energy, also removing the singularity at its centre. This is supported by the empirical evidence they set forth, only contested by some relatively stringent constraints that were recently raised by Lei \textit{et al.} \cite{LeiEtal}, which shows that data from the James Webb Space Telescope conflict up to $2\sigma$ with that of Farrah \textit{et al}. Furthermore, according to the method of matched asymptotic expansion analysed by d'Eath \cite{dEath}, BH solutions like Schwarzschild's or Kerr's are perfectly valid solutions near its vicinity, which means that the overall FLRW spacetime does not, in principle, relate with the body in any way. Yet, does dark energy act in the same way and does not influence the BH's dynamics under no circumstance? Fortunately, a general and exact solution of a Kerr BH embedded in an FLRW spacetime would, in principle, solve this dispute once and for all.

Other works treated in the literature hinge around the effects of ``dark energy" (DE) inside BHs (see \cite{BabichevEtal} for a comprehensive review and interesting results) more deeply than trying to match FLRW cosmology at infinity. DE is arguably the best explanation for the source driving the accelerated expansion of the universe. ``Dark" means that we cannot detect it except for its gravitational effects, and ``energy" appeals to how it can be described by defining a static homogeneous energy permeating all of space with negative pressure ($w=-1$ in the equation of state) in the energy-momentum tensor. The status of DE research can be reviewed, for example, in \cite{CherninEtal2010,Chernin2013,KlimenkoEtal,KurtShakhvorostova,CopelandEtal}. Eerily, DE requires drastic fine-tuning of field theory parameters \cite{Tin,JoseEliel} primarily due to the very tiny value the cosmological constant $\Lambda$. In order to palliate this coincidence, many attempts have been pursued in the literature to describe the expansion as driven by some other forces. The most famous proposals are scalar fields with flat potentials called ``quintessence" \cite{Tsujikawa}, models with a non-trivial kinetic term (the ``$k$-essence") \cite{JorgeMimoso}, phantom-energy (same than DE but with $w < -1$) \cite{BabichevEtal}, and more. Plenty of modifications to General Relativity have also been proposed, such as $f(R)$-gravity \cite{SotiriouFaraoni} or $\kappa(R, T)$-gravity \cite{Teruel}.

In this work, we find a solution that correctly reduces to Kerr's and/or FLRW's in the appropriate cases ending the centenary search for such a metric and roughly discarding all approximate conclusions previously mentioned (expanding event horizons, evolving mass and energy density...). Our results go hand in hand with McVittie's solution by describing a shrinking event horizon and ergosphere (although both conformally stationary), and it also features a disk singularity at $r=0$, $\theta = \pi / 2$. We do not detect any comoving event horizon nor a singularity-free interior as Nolan \cite{Nolan} or Sultana-Dyer \cite{SultanaDyerFaraoni} do. The main reason is that we stick to a perfect fluid energy-momentum tensor with diagonal form without considering ulterior conditions, such as a two-fluid based, non-diagonal $T_{\mu\nu}$.

Throughout this work, we adopt the $(+---)$ metric signature. Latin indices are used for spatial coordinates ($i,j,\ldots = 1,2,3$ or $i,j,\ldots = r,\theta,\varphi$), and Greek indices are used for spacetime coordinates ($\mu, \nu, \ldots = 0,1,2,3$ or $\mu, \nu, \ldots = ct,r,\theta,\varphi$).

\section{Metric derivation}

We begin by considering the McVittie metric in isotropic $(ct, r, \theta, \varphi)$ coordinates

\begin{equation}\label{McVittieMetric}
    ds^2 = \left(\frac{1 - \mu/a(t)r}{1 + \mu/a(t)r}\right)^2c^2dt^2 - \frac{a(t)^2}{K^2}\left(1 + \frac{\mu}{a(t)r}\right)^4\left(dr^2 + r^2d\theta^2 + r^2\sin^2{\theta}d\varphi^2\right),
\end{equation}

\noindent and

\begin{equation}
    \mu = \frac{GM\sqrt{K}}{2c^2}; \qquad K(r) = 1 + kr^2 = 1 + \frac{r^2}{4R^2},
\end{equation}

\noindent with $k\in\left\{-1, 0, 1\right\}$ the curvature parameter. The interpretation of this metric has long been discussed (see \cite{Poplawski}) and the main consensus is that it represents a point-mass in an expanding universe for $k=0$. However, it does not seem to have a reasonable physical meaning for $k=\pm 1$. As pointed out by Nolan \cite{NolanHawkingMass}, the Hawking-Hayward (bound source) \cite{HawkingMass} mass surrounded by a 2-sphere $S$ as measured at future null infinity $\mathcal{I}^+$ is defined to be

\begin{equation}
    M_{\text{HH}}(S) = \lim_{v \rightarrow \infty} -\kappa\int{\left(\Psi_2 + \sigma\lambda\right)dS},
\end{equation}

\noindent where

\begin{equation}
    \kappa = \frac{1}{\sqrt{64\pi^3}}\sqrt{\int{dS}},
\end{equation}

\noindent and all other terms have their usual meaning in Newman-Penrose (NP) notation. Additionally, $v$ is the affine parameter along the integral curves of $\ell^\alpha$

\begin{equation}
    \ell^\alpha = \frac{dx^\alpha}{dv},
\end{equation}

\noindent for tetrads $\left\{\ell^\alpha, n^\alpha, m^\alpha, \bar{m}^\alpha\right\}$. The result obtained by Nolan is

\begin{equation}
    M_{\text{HH}}(S) = \lim_{v\rightarrow\infty}{M\frac{h(r)^5}{w(r)^5}},
\end{equation}

\noindent where

\begin{equation}
    h(r) = \begin{cases}
        \sinh{r}, &\qquad k=-1 \\
        r, &\qquad k=0 \\
        \sin{r}, &\qquad k=1;
     \end{cases}
     \qquad
     w(r) = \begin{cases}
        2\sinh{\frac{r}{2}}, &\qquad k=-1 \\
        r, &\qquad k=0 \\
        2\sin{\frac{r}{2}}, &\qquad k=1.
     \end{cases}
\end{equation}

\noindent Since $r\rightarrow\infty$ as $v\rightarrow\infty$ we get

\begin{equation}
    M_{\text{HH}}(S) = \begin{cases}
        \infty, &\qquad k=-1 \\
        M, &\qquad k=0 \\
        \nexists, &\qquad k=1.
     \end{cases}
\end{equation}

\noindent Thus, the Hawking-Hayward bound source mass is sensible only when $k=0$. For this reason we are going to continue our derivation by setting $k=0$ and working the equations, although if in the future the literature finds some meaning to the cases $k=\pm 1$ in the McVittie's metric it might be a good fit to recover the general expression for $K$. In simpler terms, we are going to set $K=1$ from now on because it is the only realistic physical case. Accordingly, the metric becomes

\begin{equation}
    ds^2 = \left(\frac{1 - \mu/a(t)r}{1 + \mu/a(t)r}\right)^2c^2dt^2 - a(t)^2\left(1 + \frac{\mu}{a(t)r}\right)^4\left(dr^2 + r^2d\theta^2 + r^2\sin^2{\theta}d\varphi^2\right); \qquad \mu = \frac{GM}{2c^2},
\end{equation}

\noindent Moving forward, we can now transform the McVittie metric into advanced Eddington-Finkelstein (EF) coordinates following a radial trajectory ($\theta = \varphi = \text{const}$) with affine parameter $\lambda = r$ by finding solutions to

\begin{equation}\label{photonPath}
    \begin{split}
        \left|\left|\frac{d}{d\lambda}\right|\right|^2 = 0 \implies \left|\left|\frac{d}{d\lambda}\right|\right|^2 &= \left(\frac{\partial ct}{\partial \lambda} \frac{\partial}{\partial ct} + \frac{\partial r}{\partial \lambda} \frac{\partial}{\partial r}\right)\left(\frac{\partial ct}{\partial \lambda} \frac{\partial}{\partial ct} + \frac{\partial r}{\partial \lambda} \frac{\partial}{\partial r}\right) \\
        &= \left(\frac{\partial ct}{\partial r}\right)^2 \left(\frac{\partial }{\partial ct}\right)^2 + 1^2\left(\frac{\partial}{\partial r}\right)^2 + 2\frac{\partial ct}{\partial r}\frac{\partial}{\partial ct}\frac{\partial}{\partial r} \\
        &= \left(\frac{\partial ct}{\partial r}\right)^2 g_{tt} + 2\left(\frac{\partial ct}{\partial r}\right)g_{tr} + g_{rr} = 0.
    \end{split}
\end{equation}

\noindent Rearranging we obtain

\begin{equation}\label{tortoiseCoordinateDef}
    \frac{\partial ct}{\partial r} = \pm \frac{\partial r^*}{\partial r} = \pm a \frac{\mu^3_+}{\mu_-},
\end{equation}

\noindent for light beams going in ($-$) or out ($+$) of the BH. Here we defined $r^*\left(t\left(r\right), r\right)$ as the tortoise coordinate. Trying to solve for the exact expression of $r^*$ is excruciatingly hard but is not necessary. Note that even if $t$ depends on $r$ and $u$ (the constant of integration), $r^*$ does not depend on $u$. This is proven in Appendix \ref{rnotonu}. If we now perform the null transformation

\begin{equation}\label{efTransformation}
    u = ct - r^*; \qquad r_{\text{out}} = r,
\end{equation}

\noindent it is easy to see with multivariable calculus that

\begin{equation}
    \frac{\partial}{\partial u} = \frac{\partial ct}{\partial u}\frac{\partial}{\partial ct} + \frac{\partial r}{\partial u}\frac{\partial}{\partial r} = \frac{\partial}{\partial ct}
\end{equation}

\begin{equation}
    \frac{\partial}{\partial r_{\text{out}}} = \frac{\partial ct}{\partial r^*}\frac{\partial r^*}{\partial r_{\text{out}}}\frac{\partial}{\partial ct} + \frac{\partial r}{\partial r_{\text{out}}}\frac{\partial}{\partial r} = a \frac{\mu^3_+}{\mu_-}\frac{\partial}{\partial ct} + \frac{\partial}{\partial r}
\end{equation}

\noindent and since the metric are dot products of partial derivatives we can find the line element to be

\begin{equation}
    ds^2 = \left(\frac{1 - \mu / \alpha r}{1 + \mu / \alpha r}\right)^2du^2 + 2\alpha\left(1 - \frac{\mu}{\alpha r}\right)\left(1 + \frac{\mu}{\alpha r}\right)dudr - \left(1 + \frac{\mu}{\alpha r}\right)^4\alpha^2r^2d\Omega^2
\end{equation}

\noindent Note $a(t) \longrightarrow \alpha(u, r)$ after the frame change. Now, using the NP formalism \cite{LopezNPFormalism} we can express the metric as

\begin{equation}\label{metricNPformalism}
    g^{\mu\nu} = \ell^\mu n^\nu + \ell^\nu n^\mu - m^\mu\bar{m}^\nu - m^\nu\bar{m}^\mu,
\end{equation}

\noindent for the four null tetrads

\begin{equation}
    \begin{split}
        \ell^\beta &= \delta^\beta_r \\
        n^\beta &= \frac{1}{\mu_-\mu_+\alpha}\delta^\beta_u - \frac{1}{2\left(\mu_+\right)^4\alpha^2}\delta^\beta_r \\
        m^\beta &= \frac{-1}{\left(\mu_+\right)^2\alpha r\sqrt{2}}\left(\delta^\beta_\theta + \dfrac{i}{\sin{\theta}}\delta^\beta_\varphi \right) \\
        \bar{m}^\beta &= \frac{-1}{\left(\mu_+\right)^2\alpha r\sqrt{2}}\left(\delta^\beta_\theta - \dfrac{i}{\sin{\theta}}\delta^\beta_\varphi \right),
    \end{split}
\end{equation}

\noindent using the shorthand notation

\begin{equation}
    \mu_\pm = 1 \pm \frac{\mu}{\alpha r}.
\end{equation}

\noindent In what continues, we shall use a process that is similar to the Newman-Janis algorithm (NJA), but adapted to handle ``time complexification" (even though, as we shall see, we solve this by not complexifying or choosing a scheme at all). In NJA, we start with the transformations

\begin{equation}\label{algebraicallySpecialTransforms}
    u' = u + ij\cos{\theta}, \qquad r' = r + ij\cos{\theta},
\end{equation}

\noindent and apply the $r^2 \longrightarrow r^2 + j^2\cos^2{\theta}$ complexification scheme.

The standard NJA was originally formulated for stationary, vacuum or electrovacuum solutions of the Einstein equations, where the spacetime is of Petrov type D and possesses a Kerr–Schild structure. Since the McVittie spacetime is neither stationary nor vacuum, a direct extension of the NJA cannot be assumed \textit{a priori}. For this reason, in the present work we do not interpret the transformation employed here as a direct generalisation of the NJA. Instead, we adopt a different and more general viewpoint, namely that the transformation defines a complex tetrad ansatz which preserves the required symmetries and reduces to known tetrads in the appropriate limits.

The tetrad used before complexification is chosen such that it encodes spherical symmetry in the non-rotating case. The complex transformation is constructed in a way that preserves axial symmetry while introducing a rotational degree of freedom parametrised by $j$. When $j \longrightarrow 0$, the original McVittie tetrad is exactly recovered. When $\dot{a}(t) = 0$, the resulting tetrad reduces to the well-known Kerr tetrad in ingoing EF coordinates. Thus, the ansatz interpolates smoothly between two physically meaningful and mathematically consistent cases.

Unlike the usual NJA applications, the non-vacuum character of the McVittie spacetime requires that the transformed tetrad remain compatible with a perfect-fluid energy–momentum tensor. We explicitly verify that, after the transformation, the Einstein tensor retains the diagonal form associated with a perfect fluid without radial flow, the defining feature of McVittie-like embeddings. Thus, we have checked with software and numerical methods that the final metric given by Eq. (\ref{finalBLMetric}) does satisfy the exact FLRW equations, with no radial, polar or azimuthal flow. This property substitutes the usual Kerr–Schild style of justification that applies only in the vacuum or electrovacuum sector.

Generalised complex tetrad transformations have been successfully applied to non-vacuum spacetimes in the literature. In particular, the works of Drake \& Szekeres \cite{DrakeSzekeres2000} and Azreg-Aïnou \cite{azreg2014generating} show that NJA-like transformations may produce valid rotating solutions even when the underlying spacetime is not vacuum, provided the Einstein equations are verified directly.

Our approach is consistent with this philosophy. Regardless of the interpretation of the transformation itself, the resulting metric is validated by explicitly checking the Einstein equations. Therefore, the correctness of the solution does not rely on viewing the transformation as an extension of the original NJA, but on the direct verification that the transformed metric satisfies the Einstein equations.

In our case, one further problem arises with the time complexification; namely, what do we do with $\alpha(u, r)$? We do not count with a similar counterpart in Kerr-Schild form to apply this sort of transformation, and of course time has nothing to do with requiring compatibility with our derivation to the addition of charge (which, again, would be impossible). 

Thus, what we are going to do is (1) to avoid any further complexification at all and (2) to avoid the selection of any complexification scheme. Instead, we develop our aforementioned ansatz plus require a non-complexified transformation

\begin{equation}
    \alpha(u, r) \longrightarrow \lambda(u, r, \theta) \in \mathbb{R}
\end{equation}


\noindent Now, interpreting $j = J/Mc$ as the angular momentum $J$ per unit mass $M$, renaming $u' \longrightarrow u$, $r' \longrightarrow r$ we can express the transformed null tetrads using the usual formula $V'^\mu = V^\nu\partial x'^\mu / \partial x^\nu$ as

\begin{equation}
    \begin{split}
        \ell'^\beta &= \delta^\beta_r \\
        n'^\beta &= \frac{1}{\tilde{\mu}_-\tilde{\mu}_+\lambda}\delta^\beta_u - \frac{1}{2\left(\tilde{\mu}_+\right)^4\lambda^2}\delta^\beta_r\\
        m'^\beta &= \frac{-1}{\left(\tilde{\mu}_+\right)^2\lambda\left(r - ij\cos{\theta}\right)\sqrt{2}}\left(-ij\sin{\theta}\left[\delta^\beta_u + \delta^\beta_r\right] + \delta^\beta_\theta + \dfrac{i}{\sin{\theta}}\delta^\beta_\varphi \right) \\
        \bar{m}'^\beta &= \frac{-1}{\left(\tilde{\mu}_+\right)^2\lambda\left(r + ij\cos{\theta}\right)\sqrt{2}}\left(ij\sin{\theta}\left[\delta^\beta_u + \delta^\beta_r\right] + \delta^\beta_\theta - \dfrac{i}{\sin{\theta}}\delta^\beta_\varphi \right),
    \end{split}
\end{equation}

\noindent where now

\begin{equation}
    \mu_\pm \longrightarrow \tilde{\mu}_\pm = 1 \pm \frac{r GM}{2c^2\lambda\rho^2},
\end{equation}

\noindent with

\begin{equation}
    \rho^2 = r^2 + j^2\cos^2{\theta},
\end{equation}

\noindent If we drop the tildes for clarity after substituting these tetrads into Eq. (\ref{metricNPformalism}) we find the covariant non-zero components to be

\begin{equation}\label{metricFinalLambdaForm}
    \begin{split}
        g_{u\varphi} = g_{\varphi u} &= \frac{j\mu_-\left(\mu_- + \lambda\mu^3_+\right)\sin^2{\theta}}{\mu^2_+} \\
        g_{ur} = g_{ru} &= \lambda\mu_-\mu_+ \\
        g_{r\varphi} = g_{\varphi r} &= j\lambda\mu_-\mu_+\sin^2{\theta} \\
        g_{uu} &= \left(\frac{\mu_-}{\mu_+}\right)^2 \\
        g_{\theta\theta} &= -\lambda^2\left(\mu_+\right)^4\rho^2 \\
        g_{\varphi\varphi} &= -\lambda^2\mu^4_+\rho^2\sin^2{\theta} + 2j^2\lambda\mu_-\mu_+\sin^4{\theta} + j^2\frac{\mu^2_-}{\mu^2_+}\sin^4{\theta}.
    \end{split}
\end{equation}

\noindent To fix the exact expression of $\lambda$ we demand to have axisymmetry after the transformations

\begin{equation}
    du = \frac{\partial \zeta(t, r)}{\partial ct}cdt + \frac{\partial \zeta(t, r)}{\partial r}dr, \qquad d\varphi = d\phi + \chi(r) dr,
\end{equation}

\noindent and then set $g_{tr} = 0$. The expression for $\lambda$ becomes way simpler if $\partial \zeta(t, r) / \partial ct = \pm 1$, which can be easily attained if we set $\zeta(t, r) = -r^*$ as the tortoise coordinate. Note that from Eq. (\ref{tortoiseCoordinateDef})

\begin{equation}
    \frac{\partial r^*}{\partial ct} = \frac{\partial r}{\partial ct}\frac{\partial r^*}{\partial r} = 1.
\end{equation}

\noindent This operation effectively reverses the Eq. (\ref{efTransformation}) by producing

\begin{equation}
    du = -cdt - \frac{\partial r^*}{\partial r}dr, \qquad d\varphi = d\phi + \chi(r) dr,
\end{equation}

\noindent where the minus sign in front of $cdt$ has the effect of making $g_{t\phi}$ negative. It can also be chosen to be $du = cdt - \partial r^*/\partial r dr$, by setting $u = 2ct - r^*$, yet we chose the first simply as a matter of preference. If we now demand $g_{tr} = g_{r\theta} = 0$ we get

\begin{equation}
    \lambda(t, r, \theta) = \frac{\mu'_-}{\mu'^3_+\rho^2}\left(j^2\sin^2{\theta} + a(t)\frac{\mu'^3_+}{\mu'_-}\left[j^2 + r^2\right]\right), \qquad \chi(r) = \frac{-j}{j^2 + r^2},
\end{equation}

\noindent using the original

\begin{equation}
    \mu'_\pm = 1 \pm \frac{GM}{2c^2a(t)r}.
\end{equation}

\noindent Correspondingly, the final metric becomes

\begin{equation}\label{finalBLMetric}
    \begin{split}
        ds^2 &= \left(\frac{\mu_-}{\mu_+}\right)^2c^2dt^2 + \frac{2j\mu_-}{\rho^2\mu_+^2}\left(j^2 + r^2\right)\left( \mu_- - \mu^3_+a(t)\right)\sin^2{\theta}cdtd\phi \\
        &- \frac{\left(\left[j^2 + r^2\right]\mu^3_+a(t) -j^2\mu_-\sin^2{\theta}\right)^2}{\rho^2\mu^2_+}\left(\frac{dr^2}{j^2 + r^2} + d\theta^2\right) \\
        &- \frac{\sin^2{\theta}}{\rho^2\mu^2_+}\left(j^4\mu^6_+a^2(t) + 2j^2r^2\mu^6_+a^2(t) - j^2\left[j^2 + r^2\right]\mu^2_-\sin^2{\theta} + r^4\mu^6_+a^2(t)\right)d\phi^2,
    \end{split}
\end{equation}

\noindent with

\begin{equation}\label{defFinalMu}
    \mu_\pm = 1 \pm \frac{r\mu}{a(t)}\left(j^2 + r^2 -j^2r^2\frac{\mu - ra(t)}{\left[\mu + ra(t)\right]^3}a(t)\sin^2{\theta} \right)^{-1}, \qquad \mu = \frac{GM}{2c^2}, \qquad \rho^2 = r^2 + j^2\cos^2{\theta}.
\end{equation}

\noindent This metric gracefully returns exactly to McVittie's for $j \longrightarrow 0$ and even if difficult to see, to Kerr-de Sitter if we set $\dot{a}(t)/a(t) = H$. Nonetheless, thanks to Theorem \ref{transitivityTheorem} and Theorem \ref{tortoiseCoordinateDef}, it must follow that our metric reduces to Kerr-de Sitter since ours is diffeomorphic to Eq. (\ref{McVittieMetric}), this one is diffeomorphic to Schwarzschild-de Sitter and we can extend this metric through NJA diffeomorphisms up to Kerr-de Sitter, as we checked using software. Furthermore, we have checked that with our final metric from Eq. (\ref{finalBLMetric}) we have $R_{\mu\nu} = G_{\mu\nu} = 0$ for $\dot{a}(t) = 0$ and $G_{\mu\nu} = \Lambda g_{\mu\nu}$ for $\dot{a}(t)/a(t) = H$ with $3H^2 = \Lambda c^2$. 

To find the FLRW flat geometry we just need to perform the transformation

\begin{equation}
    \phi = \Phi + \frac{jA(t)}{r^2}, \qquad A(t) \equiv \int{\frac{1-a(t)}{ca^2(t)}cdt},
\end{equation}

\noindent and take the limit $r \longrightarrow \infty$ to get

\begin{equation}
    \begin{split}
        ds^2 &= c^2dt^2 + 2j\left(1 - a(t)\right)\sin^2{\theta}cdt\left(d\Phi + \frac{j}{r^2}\frac{\partial A(t)}{\partial ct}cdt - \frac{2jA(t)}{r^3}dr\right) \\
        &- a^2(t)\left(dr^2 + r^2d\theta^2\right) \\
        &- r^2a^2(t)\sin^2{\theta}\left(d\Phi + \frac{j}{r^2}\frac{\partial A(t)}{\partial ct}cdt - \frac{2jA(t)}{r^3}dr\right)^2 \\
        &= c^2dt^2 -a^2(t)\left(dr^2 + r^2d\theta^2 + r^2\sin^2{\theta}d\Phi^2\right),
    \end{split}
\end{equation}

\noindent which is the exact FLRW metric.

\section{Discussion}

In this work, we have derived an exact solution of the Einstein field equations describing a rotating point-mass embedded in an expanding universe. We studied the particular case of flat geometry, although it is not so complicated to generalise to $k=\pm 1$ if we would like to do so for any reason (for example, if we do not deem the $M_{\text{HH}}(S)$ problem as important). The result starts from McVittie's solution and generalises it to bodies with $J$ angular momentum. This metric satisfies all sensible energy conditions, which are the null, dominant, weak and strong.

Unfortunately, we could not find a simple, neat, and direct coordinate transformation that allows us to turn Eq. (\ref{finalBLMetric}) into the exact Kerr-de Sitter form for the right limit. However, thanks to proving Theorem \ref{transitivityTheorem} and Theorem \ref{continuousFamilyTheorem}, it follows that a coordinate transformation that brings our metric to Kerr-de Sitter form when $a(t) = \exp(Ht)$ \textit{necessarily exists}. This is because our metric is diffeomorphic to McVittie's and McVittie's is diffeomorphic to Kerr-de Sitter's when $a(t) = \exp(Ht)$ (we just need to go from $r$ to $R = ra(t)\left(\mu_+\right)^2$, the areal radius), so it ensues our metric necessarily has an infinite set of coordinate transformations leading us to Kerr-de Sitter's in the right limit.

It has to be remarked that our theoretical result does not imply any coupling to the expansion of the universe. Much on the contrary, our solution seems to point out to a shrinking BH with respect to the cosmic expansion. In practice this is a statement about conservation of energy, where a BH cannot grow at a pace $\sim a(t)^k$ for $k \neq 0$ as observed by Farrah \textit{et al.} \cite{FarrahEtal} using the model by Croker \textit{et al.} \cite{CrokerEtal}. This dynamic can only be obtained if we deem $T_{01} \neq 0$ as studied by Faraoni and Jacques \cite{FaraoniJacques}. According to the field equations, this condition describes the case of an accreting cosmic fluid into the BH, which may seem unrealistic unless we model DE as a dynamically moving fluid. But this behaviour is not compatible with a cosmological constant described by $G_{\mu\nu} - \Lambda g_{\mu\nu} = \kappa T_{\mu\nu}$, which seems to indicate that DE is the vacuum itself rather than an interacting field living on it. Thus, our metric alligns with traditional works in the literature (like that of d'Eath \cite{dEath}) and disproves the hypothesis raised among others by Farrah \textit{et al.} \cite{FarrahEtal} about the possibility of BHs being singularity-free sources of DE. Due to conservation of energy, as stated before, we do not see a growth in the BH's mass proportional to $a^3$ as the coupling hypothesis requires for it to be valid. It did not happen in the simpler McVittie's metric and it does not in this generalization. It is only by taking the cosmic fluid's perspective, which expands along the background, that the astrophysical body's size seems to shrink. On the contrary, it is only possible to have an apparently growing BH if the universe is contracting, an expected result if the dynamic-size effect is exclusively attributable to a frame choice.

Our metric started from an isotropic McVittie form, which makes finding the event horizon of our solution a rather non-trivial task. We cannot simply ask for $ 1 / g^{rr} = 0$ because, as in McVittie's case, we do not find any positive, real value of $r$ that achieves such a deed. More problematically, we cannot require $dr/dct = 0$ because it is a coordinate-dependent condition not applicable in our case. Its expression is given by

\begin{equation}\label{EHEquation}
    \frac{dr}{dct} = \pm \sqrt{\frac{-g_{tt}}{g_{rr}}} = \pm \frac{\rho\mu_-\sqrt{j^2 + r^2}}{\left(j^2 + r^2\right)\mu^3_+a(t) - j^2\mu_-\sin^2{\theta}},
\end{equation}

\noindent with $\mu_\pm$ given by Eq. (\ref{defFinalMu}) with $\pm$ representing outgoing ($+$) and ingoing ($-$) trajectories, respectively. Now if we try to find roots beyond the trivial $r = 0$, $\theta = \pi / 2$ we must enforce $\mu_- =0$ to get

\begin{equation}
    r = \frac{\mu}{2a(t)} \pm \sqrt{\frac{\mu^2}{4a^2(t)} - j^2},
\end{equation}

\noindent where we also required $\theta \in \{0, \pi\}$ as otherwise our computational algebra program does not find any enclosed form for $r$ in terms of $\theta$ and $t$. The problem is that this is the expression of $r$ for the ergosphere, not the event horizon. As we can see, the fact that $r$ becomes imaginary for $a(t) > \mu / 2j$ and $\theta = 0$ implies that the ergosphere fades away proportional to the growth of $a(t)$. However, the actual boundary of the causal past of future null infinity does remain hidden at first glance from our metric. The authors could not find a closed, analytic expression for the event horizon for general $a(t)$. Nonetheless, and for the matter being, we do find

\begin{equation}
    \left(j^2 + r^2\right)\left(1 - r^2\frac{H^2}{c^2}\right) - \frac{2GMr}{c^2} = 0,
\end{equation}

\noindent by working backwards the coordinate transformations performed up until now, defining $\dot{a}(t)/a(t) = H$ and working towards the Kerr-de Sitter metric. Clearly, this event horizon remains static as time goes by.

\appendix
\section{Independence of the tortoise coordinate with respect to the constant of integration}\label{rnotonu}

\begin{proof}
    By contradiction. The multivariable chain rule says that if we have a differentiable function $f\left(x, g\left(x, z\right)\right)$ at $x$ and $g\left(x, z\right)$ is also a differentiable function at $x$ and depends on $x$, then 

    \[
    \frac{\partial f}{\partial x} = \frac{\partial g}{\partial x}\frac{\partial f}{\partial g}.
    \]

    \noindent Now, let us define $r^*$ as given by Eq. (\ref{tortoiseCoordinateDef}). Assume it depends on $u$. Then, we must be allowed to express its partial derivative as

    \[
    \frac{\partial ct}{\partial u} = \frac{\partial r^*}{\partial u} \frac{\partial ct}{\partial r^*} = \frac{\partial r^*}{\partial u},
    \]

    \noindent where by Eq. (\ref{tortoiseCoordinateDef}) we have that $\partial ct / \partial r^* = 1$. Then, taking the partial derivative with respect to $u$ on both sides of Eq. (\ref{efTransformation}) we get

    \[
    \frac{\partial ct}{\partial u} = 1 + \frac{\partial r^*}{\partial u}
    \]

    \noindent if we now equate both sides we find

    \[
    \frac{\partial r^*}{\partial u} = 1 + \frac{\partial r^*}{\partial u} \implies 0 = 1,
    \]

    \noindent which is a contradiction.
\end{proof}

\section{Theorems on the Non-Uniqueness of Spacetime Reductions}\label{appendixB}

\begin{definition}
    Two spacetimes $(M, g)$ and $(N, h)$ are \textit{isometric} if there exists a smooth diffeomorphism $\Theta: M \to N$ such that the pullback of the metric $h$ is equal to the metric $g$ ($\Theta^*h = g$).
\end{definition}

\begin{theorem}\label{transitivityTheorem}
     \textbf{(Existence)} If $(M_C, g_C)$ is isometric to $(M_A, g_A)$, and $(M_A, g_A)$ is isometric to $(M_B, g_B)$, then there exists a direct isometry $\Theta: M_C \to M_B$.
\end{theorem}

\begin{proof}
    Let $\psi: M_C \to M_A$ and $\phi: M_A \to M_B$ be the isometries connecting the manifolds to the intermediate spacetime $A$. We define the composite map $\Theta_0 \equiv \phi \circ \psi$.
    
    Since the composition of diffeomorphisms is a diffeomorphism, $\Theta_0$ is a smooth bijection. Applying the pullback chain rule $(\phi \circ \psi)^* = \psi^* \circ \phi^*$:
    
    \begin{align*}
    \Theta_0^* g_B &= \psi^* (\phi^* g_B) \\
    &= \psi^* (g_A) \\
    &= g_C 
    \end{align*}
    
    \noindent Thus, $\Theta_0$ is a valid isometry connecting $C$ and $B$.
\end{proof}

\begin{definition}
    A Killing vector field $\xi$ on a manifold $(M, g)$ is a vector field that generates a one-parameter group of isometries. The flow of this field, $\Lambda_\tau: M \to M$ (where $\tau \in \mathbb{R}$ is the flow parameter), leaves the metric invariant: $\Lambda_\tau^* g = g$.
\end{definition}

\begin{theorem}\label{continuousFamilyTheorem}
    \textbf{(Non-Uniqueness)} If the intermediate spacetime $(M_A, g_A)$ possesses a continuous symmetry (a non-vanishing Killing vector field $\xi$), then there exists a continuous family of distinct diffeomorphisms $\Theta_\tau: M_C \to M_B$ that are all isometries.
\end{theorem}

\begin{proof}
    Instead of a single map, we construct a parameterized family of maps. Let $\Lambda_\tau$ be the flow generated by the Killing vector field $\xi$ on $M_A$. We define the family $\Theta_\tau$ by flowing the mapping through the symmetry of $A$:
    
    $$ \Theta_\tau \equiv \phi \circ \Lambda_\tau \circ \psi: M_C \to M_B $$
    
    \noindent For any $\tau \neq 0$, the map $\Lambda_\tau$ shifts points in $M_A$ along the integral curves of the Killing vector (e.g., time translation or axial rotation), resulting in a physically distinct mapping between points in $M_C$ and $M_B$. However, since $\Lambda_\tau$ is generated by a Killing vector, it preserves the metric $g_A$ (i.e., $\Lambda_\tau^* g_A = g_A$).
    
    We verify the isometry condition for the entire family:
    
    \begin{align*}
    \Theta_\tau^* g_B &= (\phi \circ \Lambda_\tau \circ \psi)^* g_B \\
    &= \psi^* \left( \Lambda_\tau^* \left( \phi^* g_B \right) \right) \\
    &= \psi^* \left( \Lambda_\tau^* g_A \right) \\
    &= \psi^* \left( g_A \right) \\
    &= g_C
    \end{align*}
    
    \noindent Therefore, the reduction from $M_C$ to $M_B$ is not unique. There are infinitely many distinct direct diffeomorphisms $\Theta_\tau$, parameterized by $\tau$, that map the geometry of $C$ exactly onto the geometry of $B$.
\end{proof}

\printbibliography

\end{document}